\documentclass[a4paper,12pt]{article} 
\usepackage[onehalfspacing]{setspace}
\usepackage{color}
\usepackage{enumitem}
\usepackage{amsmath,amssymb,amsthm}
\usepackage{subcaption}
\usepackage{graphicx}
\usepackage{algorithm}
\usepackage{algpseudocode}
\usepackage{color}
\usepackage{anysize}
\usepackage{multirow}
\usepackage{authblk}

\newtheorem{theorem}{Theorem}[section]
\newtheorem{lemma}[theorem]{Lemma}
\newtheorem{corollary}[theorem]{Corollary}

\title{Liar's Domination in Unit Disk Graphs\thanks{Preliminary version of this paper appeared in COCOON, 2018}}

\author[1]{Ramesh K. Jallu \thanks{jallu@cs.cas.cz}\thanks{Ramesh K. Jallu was supported by the Czech Science Foundation, grant number GJ19-06792Y, and by institutional support
RVO:67985807}}
\author[2]{Sangram K. Jena \thanks{sangram@iitg.ac.in}}
\author[2]{Gautam K. Das \thanks{gkd@iitg.ac.in}}
\affil[1]{The Czech Academy of Sciences, Institute of Computer Science}
\affil[2]{Department of Mathematics\\Indian Institute of Technology Guwahati }
\begin{document}

\maketitle

\begin{abstract}
In this article, we study a variant of the minimum dominating set problem known as the  
minimum liar's dominating set (MLDS) problem. We prove that the MLDS problem is 
NP-hard in unit disk graphs. Next, we show that the recent sub-quadratic time 
$\frac{11}{2}$-factor approximation algorithm \cite{bhore} for the MLDS problem 
is erroneous and propose a simple $O(n + m)$ time 7.31-factor approximation algorithm, 
where $n$ and $m$ are the number of vertices and edges in the input unit disk graph, 
respectively. Finally, we prove that the MLDS problem admits a polynomial-time approximation scheme.

\noindent{\bf keywords:} Dominating set, Liar's dominating set, Unit Disk Graph, Approximation scheme
\end{abstract}
\section{Introduction}
Given a simple undirected graph $G=(V,E)$, the open and closed neighborhoods
of a vertex $v_i \in V$ are defined by
 $N_G(v_i) = \{v_j\in V \mid (v_i,v_j)\in E\ and\ v_i \neq v_j\}$ 
and $N_G[v_i] = N_G(v_i) \cup \{v_i\}$, 
respectively. A \emph{dominating set} $D$ of $G$ is a subset of $V$ such
that every vertex in $V\setminus D$ is adjacent to at least one
vertex in $D$. That is, each vertex $v_i \in V$ is either in $D$ or
there exists a vertex $v_j\in D$ such that $(v_i,v_j)\in E$. Observe that 
for any dominating set $D \subseteq V$, $|N_G[v_i]\cap D| \geq 1$ for each $v_i \in V$. 
We say that a vertex $v_i$ is dominated by $v_j$ in $G$, if $v_j \in D$ and $(v_i,v_j)\in E$.
The dominating set problem asks to find a dominating set of minimum size in a given graph.
A set $D\subseteq V$ is a {\it $k$-tuple dominating set} in $G$, if each
vertex $v_i \in V$ is dominated by at least $k$ vertices in $D$. In other words, 
$|N_G[v_i]\cap D| \geq k$ for each $v_i \in V$. 
The minimum cardinality of a $k$-tuple dominating set of a graph $G$ is called the 
\emph{$k$-tuple domination number} of $G$.

A \emph{liar's dominating set} (LDS) in a simple undirected graph $G=(V,E)$, is a 
dominating set $D$ having the following two properties: (i) for every
$v_i \in V$, $|N_G[v_i]\cap D| \geq 2$, and (ii) for every
pair of distinct vertices $v_i$ and $v_j$ in $V$, $|(N_G[v_i]\cup N_G[v_j])\cap D| \geq 3$.
For a given graph $G$, the problem of finding an LDS  
in $G$ of minimum cardinality is known as the \emph{minimum liar's 
dominating set} (MLDS) problem. The cardinality of an MLDS in a graph $G$ 
is known as the liar's domination number of $G$. Every 3-tuple dominating set is a
liar's dominating set as it satisfies both conditions, so the liar's domination
number lies between 2-tuple and 3-tuple domination numbers.

Our interest in the LDS problem arises from the following scenario.
Consider a graph in which each node is a possible location 
for an intruder such as a thief, or a saboteur. We would like to detect and 
report the intruder's location in the graph. A protection device such as a camera or a sensor 
placed at a node can not only detect (and report) the intruder's presence at it, but also 
at its neighbors. Our objective is to place a minimum number of protection devices such that 
the intrusion of the intruder at any vertex is detected
and reported. In this situation, one 
must place the devices at the vertices of a minimum dominating set of the graph
to achieve the goal. The protection devices are prone to failure and hence 
certain degree of redundancy is needed in the solution. Also, some times the devices
 may misreport the intruder's location deliberately or due to transmission error. Assume that
at most one protection device in the closed neighborhood of the intruder can lie
(misreport). In this context, 
one must place the protection devices at the vertices of an MLDS of the graph to 
achieve the objective. 
The first property in the definition of LDS deals with single device fault-tolerance, while 
the second property deals with the case in which two distinct locations about the intruder are reported.

\section{Related Work}

The MLDS problem is introduced by Slater \cite{slater}.
He showed that the problem is NP-hard for general graphs, and gave a lower 
bound on the liar's domination number in case of trees by proving that 
the size of any liar's dominating set of a tree of order $n$ is between
$\frac{3}{4}(n+1)$ and $n$. 
Later, Roden and Slater \cite{roden}
characterized tree classes with liar's 
domination number equal to $\frac{3}{4}(n+1)$. In the same paper, they also 
showed that the MLDS problem is NP-hard even for bipartite graphs. 
Panda and Paul \cite{panda2013liar} proved that the problem is NP-hard
for split graphs and chordal graphs. They also proposed a
linear time algorithm for computing an MLDS in case of trees. 

Panda et al. \cite{panda2015hardness} studied the approximability
of the problem and presented an $O(\ln \Delta)$-factor 
approximation algorithm, where $\Delta$ is the
degree of the graph. Panda and Paul \cite{paul2013}
considered the problem for proper interval graphs and proposed
a linear time algorithm for computing a minimum cardinality
liar's dominating set. The problem is also studied for bounded
degree graphs, and $p$-claw free graphs \cite{panda2015hardness}.
Sterling \cite{sterling} considered the problem on two-dimensional 
grid graphs and presented bounds on the liar's domination number. 

Alimadadi et al. \cite{alimadadi}
provided the characterization of graphs and trees for which the liar's
domination number is $|V|$ and $|V|-1$, respectively. 
Panda and Paul \cite{panda2013connected,panda2014hardness} studied 
variants of liar's domination, namely,
connected liar's domination and total liar's domination. 
A \emph{connected liar's dominating set} (CLDS) is an LDS whose induced 
subgraph is connected. A \emph{total liar's dominating set} (TLDS) is a 
dominating set $D$ with the following two properties: (i) for every
$v \in V$, $|N_G(v)\cap D| \geq 2$, and (ii) for every distinct
pair of vertices $u$ and $v$, $|(N_G(u)\cup N_G(v))\cap D| \geq 3$, where 
$N_G(\cdot)$ is the open neighborhood of a vertex. The objective of 
both problems is to find CLDS and TLDS of minimum size, respectively.
The authors also proved that both problems are NP-hard and proposed  
$O(\ln \Delta)$-factor approximation algorithms. They also proved 
that the problems are APX-complete for graphs with maximum degree 4. 
Jallu and Das \cite{jallu2017liar} first studied the geometric version of the MLDS problem, and presented constant factor approximation algorithms with high running time. 
Recently, Banerjee and Bhore \cite{bhore} proposed a $\frac{11}{2}$-factor approximation algorithm in sub-quadratic time. However, unfortunately, their approximation analysis is erroneous and the approximation factor is at least 11 (refer Section \ref{appxalgo}).
\subsection{Our Contribution} \label{contribution}

We study the MLDS problem on a geometric intersection graph model, particularly in UDGs. 
A \emph{unit disk graph} (UDG) is an intersection graph of equal radii disks in 
the plane. Given a set $\{d_1, d_2,\ldots, d_n\}$ of $n$ circular disks in the plane, 
each having radius 1, the corresponding UDG $G = (V,E)$ is defined as follows: each vertex $v_i \in V$ corresponds to a disk $d_i$, and there is an edge between two vertices $v_i$ and $v_j$ if and only if 
the Euclidean distance between the corresponding disk centers $d_i$ and $d_j$ is at most 1. 

We show that the decision version of the MLDS problem is NP-complete in UDGs (refer to Section \ref{hardness}). We propose a simple linear time 7.31-factor approximation algorithm and a PTAS in Section \ref{appxalgo} and Section \ref{ptas}, respectively. Finally, we conclude the paper in Section \ref{conclusion}.

\section{Hardness of the MLDS Problem in UDGs}\label{hardness}
In this section, we show that the MLDS problem in UDGs is NP-complete by 
reducing the {\it vertex cover} problem 
defined in planar graphs to it, which is known to be NP-complete \cite{garey}. 
The decision versions of both the problems are formally defined below.
\begin{description}
 \item [The MLDS problem in UDGs] (\textsc{Lds-Udg})
 \item [Instance:] A unit disk graph $G =(V,E)$ and a positive integer $k$.
 \item [Question:] Does there exist a liar's dominating set  $D$ in $G$ such that $|D|\leq k$?.
\end{description}
\begin{description}
 \item [The vertex cover problem in planar graphs] (\textsc{Vc-Pla})
 \item [Instance:] A simple planar graph $G$ with maximum degree 3 and a positive integer $k$.
 \item [Question:] Does there exist a vertex cover $C$ of $G$ such that $|C|\leq k$?.
\end{description}
 
\begin{lemma}[\cite{valiant}] \label{key-lemma}
 A planar graph $G=(V,E)$ with maximum degree 4 can be embedded in the plane
 using $O(|V|^2)$ area in such a way that its vertices are at integer
 coordinates and its edges are drawn so that they are made up of line
 segments of the form $x=i$ or $y=j$, for integers $i$ and $j$. 
\end{lemma}

This kind of embedding is known as \emph{orthogonal drawing} of a graph. Biedl and 
Kant \cite{biedl1998} gave a linear time algorithm that produces an orthogonal 
drawing of a given graph with the property that the number of bends along 
each edge is at most 2 (see Figure \ref{graph_grid}).

\begin{corollary}
A planar graph $G=(V,E)$ with maximum degree 3 and $|E| \geq 2$ can be
embedded in the plane such that its vertices are at $(4i,4j)$
and its edges are drawn as a sequence of
consecutive line segments on the lines $x=4i$ or $y=4j$, for  
integers $i$ and $j$.
\end{corollary}
 
\begin{figure}[ht]
\centering
\begin{subfigure}[b]{.5\textwidth}  
  \centering
  \includegraphics[scale=0.75]{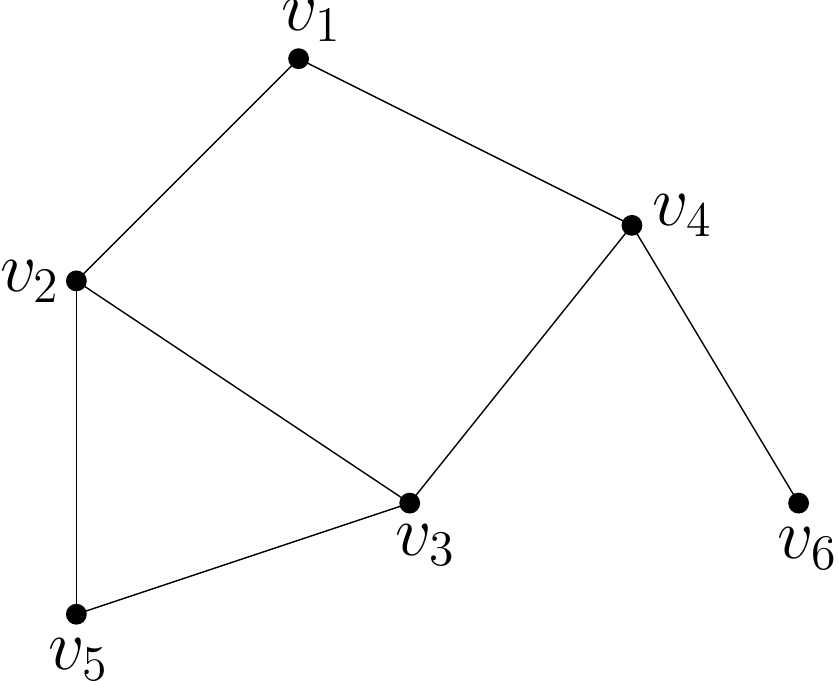}
 \caption{}
\end{subfigure}%
\hspace*{-1cm}
\begin{subfigure}[b]{.5\textwidth}  
  \centering
  \includegraphics[scale=0.75]{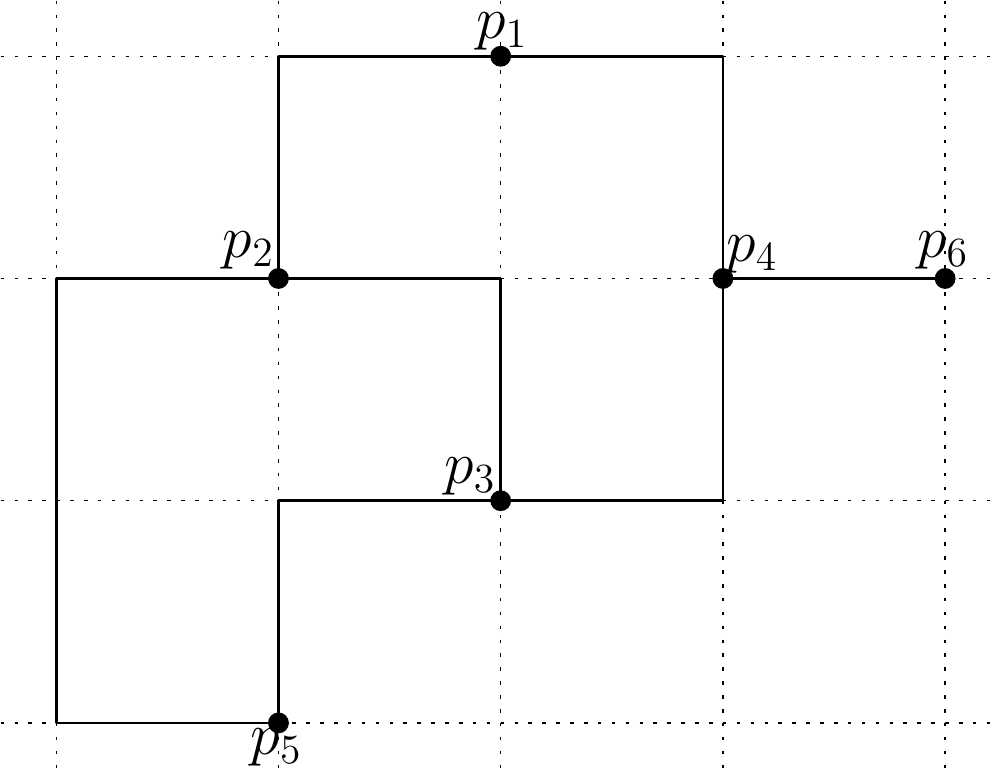}
\caption{}  
\end{subfigure}
\begin{subfigure}[b]{.5\textwidth}  
  \centering
  \includegraphics[scale=0.75]{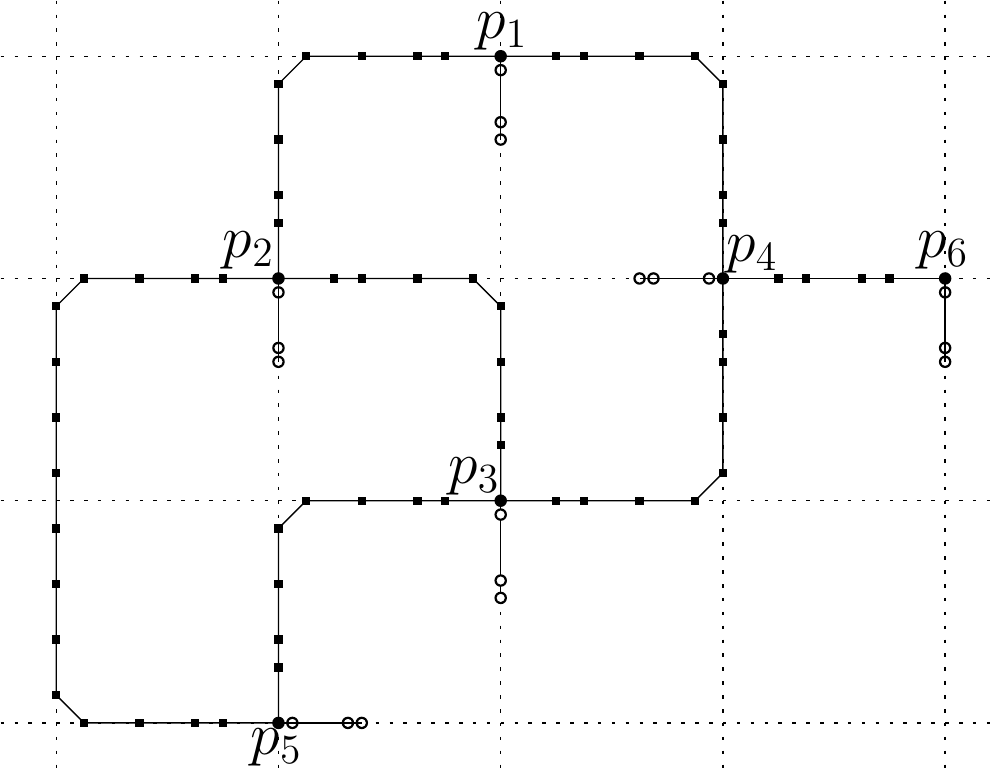}
  \caption{}  
\end{subfigure}

\caption{(a) A planar graph $G$ with maximum degree 3, 
  (b) its embedding on a grid, and (c)  construction of an UDG from the embedding.}\label{graph_grid}
\end{figure}
\begin{lemma}\label{lemma-udg}
 Let $G=(V,E)$ be an instance of \textsc{Vc-Pla} with $|E| \geq 2$. An instance $G'=(V',E')$ of 
 \textsc{Lds-Udg} can be constructed from $G$ in polynomial-time.
\end{lemma}

\begin{proof}
 We construct $G'$ in four phases. \\ 
 {\bf Phase 1: Embedding of $G$ into a grid of size $4\times 4$}\\
 Embed the instance $G$ in the plane as discussed previously 
 using one of the algorithms in \cite{hopcroft,itai}. An edge 
 in the embedding is a sequence of connected line segment(s) 
 of length four units each. If the total number of line segments
 used in the embedding is $\ell$, then the sum of the lengths of the line
 segments is $4\ell$ as each line segment has length 4 units. 
 We name the points in the embedding correspond to the vertices of $G$ 
 by \emph{node points} (see Figure \ref{graph_grid}(b)).
 
 \noindent
 {\bf Phase 2: Adding extra points to the embedding}\\
 Divide the set of line segments in the embedding into two
 categories, namely, proper and improper. 
 We call a line segment \emph{proper} if none of its end
 points correspond to a vertex in $G$. A line segment is \emph{improper} if 
 it is not a proper segment. For each edge $(p_i,p_j)$ of
 length 4 units we add two points at distances 1 and 1.5 units of $p_i$ and $p_j$, 
 respectively (thus adding four points in total, see the edge $(p_4,p_6)$ in 
 Figure \ref{graph_grid}(c)).
 For each edge of length greater than 4 units,
 we also add points as follows: for each improper line segment we add four points
 at distances 1, 1.5, 2.5, and 3.5 units from the endpoint corresponding to a
 vertex in $G$, and for each proper line segment we add four points at distances
 0.5 and 1.5 units from its endpoints (see Figure \ref{graph_grid}(c)). We name the points 
 added in this phase \emph{joint points}.\\
 
 \noindent
 {\bf Phase 3: Adding extra line segments and points}\\
 Add a line segment of length 1.4 units (on the lines $x=4i$ or $y=4j$
 for some integers $i$ or $j$) for every point $p_i$, which corresponds to a vertex $v_i$ in
 $G$, without coinciding with the line segments that had already been drawn.
 Observe that adding this line segment on the lines $x=4i$ or $y=4j$
 is possible without losing the
 planarity as the maximum degree of $G$ is 3. Now, add three points 
 (say $x_i$, $y_i$, and $z_i$) on these line segments at distances 0.2, 1.2, and 1.4 units, 
 respectively, from $p_i$.
 We name the points  added in this phase \emph{support points}.\\
 {\bf Phase 4: Construction of UDG}\\
 For convenience, let us denote the set of node points, joint points, and support 
 points by $N$, $J$, and $S$, respectively. Let $N = \{p_i \mid v_i \in V\}$,
 $J=\{q_1,q_2,\ldots,q_{4\ell}\}$, and $S=\{x_i,y_i,z_i \mid v_i \in V\}$. 
 We construct a UDG $G'=(V',E')$,
 where $V' = N \cup J \cup S$ and there is an edge between two points in $V'$
 if and only if the Euclidean distance between the points is at most 1
 (see Figure \ref{graph_grid}(c)).  
 Observe that, $|N| = |V| (= n)$, $|J| = 4\ell$, where $\ell$ is the
 total number of line segments in the embedding, and $|S| = 3|V|(=3n)$. Hence, $|V'| = 4(n+\ell)$
 and $\ell$ is bounded by a polynomial of $n$. Therefore $G'$ can be constructed in
 polynomial-time.
 \end{proof}
\begin{theorem}\label{thm:main}
 \textsc{Lds-Udg} is NP-complete.
\end{theorem}

\begin{proof}
 \textsc{Lds-Udg}$\in NP$, since for any given set $D\subseteq V$ and a positive integer $k$, we can verify  whether $D$ is a liar's dominating set of size at most $k$ or not in polynomial-time.
  
 We prove the hardness of \textsc{Lds-Udg} by reducing \textsc{Vc-Pla} to it. Let 
 $G=(V,E)$ be an instance of \textsc{Vc-Pla}. Construct an instance $G'=(V',E')$ 
 of \textsc{Lds-Udg} as discussed in Lemma \ref{lemma-udg}.
 We now prove the following claim: {\it $G$ has a vertex cover of size at most $k$
 if and only if $G'$ has a liar's dominating set of size at most $k + 3\ell + 3n$}.

 \noindent {\bf Necessity:} Let $C\subseteq V$ be a vertex cover 
 of $G$ such that $|C| \leq k$. Let $N' = \{p_i \in N \mid v_i \in C\}$, i.e., 
 $N'$ is the set of vertices (or node points) in $G'$ that correspond to the vertices in $C$.  
 From each segment in the embedding we choose 3 vertices (joint points). 
 The set of chosen vertices, say $J'(\subseteq J)$, 
 together with $N'$ and $S$ will form an LDS of desired cardinality in $G'$.
 We now discuss the process of obtaining the set $J'$. Initially $J' = \emptyset$.
 As $C$ is a vertex cover, every edge in $G$ has at least one of its end vertices in $C$. 
 Let $(v_i,v_j)$ be an edge in $G$ and $v_i \in C$ (the tie can be broken arbitrarily 
 if both $v_i$ and $v_j$ are in $C$). Note that the edge $(v_i,v_j)$ is 
 represented as a sequence of line segments in the embedding. 
 Start traversing the segments (of $(v_i,v_j)$) from $p_i$, where $p_i$ corresponds to $v_i$,
 and add all the vertices to $J'$ except the first one from each 
 segment encountered in the traversal (see $(p_2,p_5)$ in Figure \ref{fig:proof} (b). The 
 bold vertices are part of $J'$ while traversing from $p_2$). 

\begin{figure}[!ht]
\centering
\begin{subfigure}[b]{.5\textwidth}  
  \centering
  \includegraphics[width=0.65\linewidth]{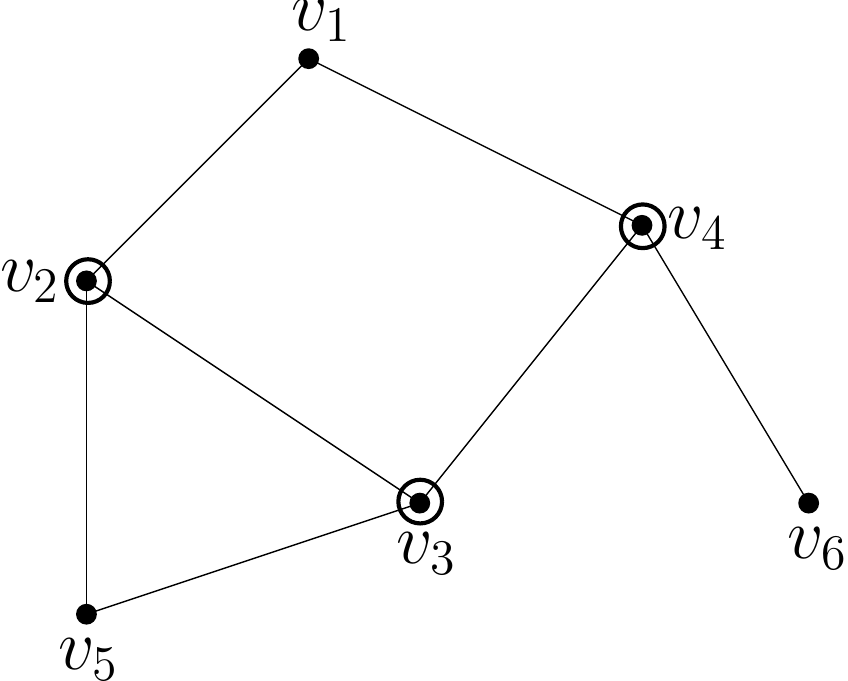}
  \caption{}  
\end{subfigure}%
\hfill
\begin{subfigure}[b]{.5\textwidth}  
  \centering
  \includegraphics[width=0.75\linewidth]{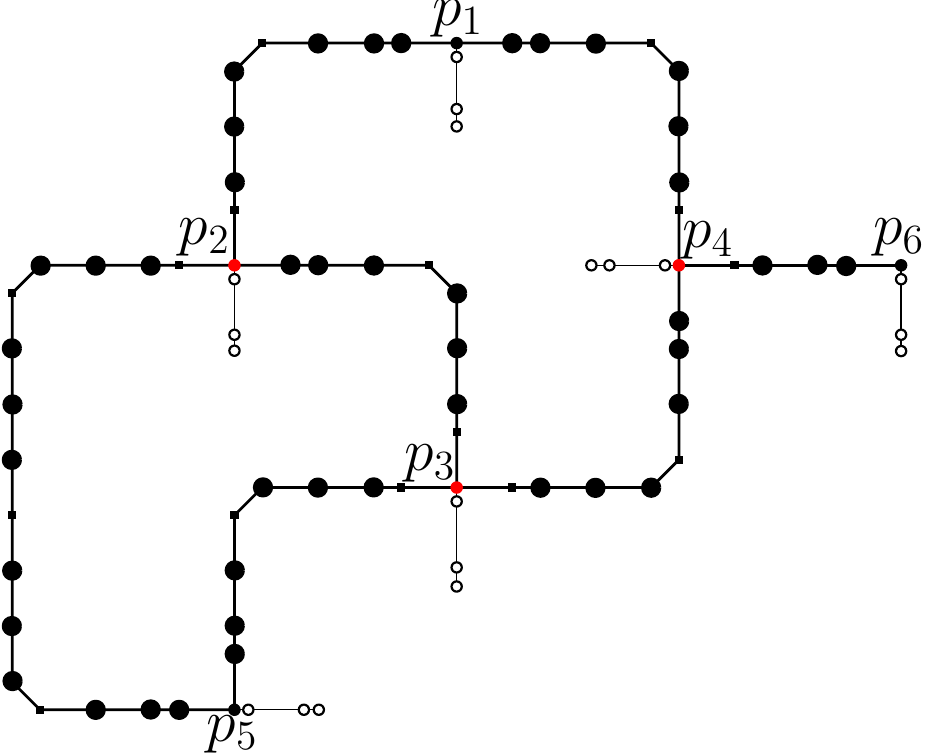}
  \caption{}  \label{fig:proofb}
\end{subfigure}
\caption{(a) A vertex cover $\{v_2,v_3,v_4\}$ in $G$, and (b) the construction 
  of $J'$ in $G'$ (the tie between $v_2$ and $v_3$, and $v_3$ and $v_4$ is broken 
  by choosing $v_3$)}\label{fig:proof}
\end{figure}
\newpage
Apply the  above process to each edge in $G$.
 Observe that the cardinality of $J'$ is $3\ell$ as we have chosen 3 vertices from each 
 segment in the embedding. Let $D = N' \cup J' \cup S$.  Now, we argue 
 that $D$ is a liar's dominating set in $G'$.
 \begin{enumerate}
 \item Each $p_i \in N$ is dominated by $x_i$ in $S$.
 If $p_i \in N'$ (i.e., the corresponding vertex $v_i \in C$ in $G$), 
 then $|N_{G'}[p_i] \cap D| \geq |\{p_i,x_i\}| = 2$. If $p_i \notin N'$,
 then there must exist at least one vertex $q_j$ in $J'$ dominating $p_i$. 
 The existence of $q_j$ is guaranteed by the way we constructed $J'$.
 Hence, $|N_{G'}[p_i] \cap D| \geq |\{q_j,x_i\}| = 2$. 
 In either case every vertex in $N$ is dominated by at least two vertices in $D$.
 It is needless to say that vertex in $J$ is dominated by at least two 
 vertices in $N' \cup J'$. Similarly, every vertex in $S$ is dominated 
 by itself, by its neighbor(s) in $S$, and, perhaps, by one vertex in $N'$. Therefore, 
 every vertex in $V'$ is double dominated by vertices in $D$.
 
 \item Consider a pair of distinct vertices in $V'$. Of course, 
 every pair of distinct vertices in $S$ satisfy the liar's second 
 condition. We prove that remaining pairs of distinct vertices 
 also do satisfy the liar's second condition by considering all possible 
 cases.\\
  {\bf Case a. $p_i, p_j \in N$:} If at least one of $p_i,p_j$ belongs to 
  $N'$ (without loss of generality say $p_i \in N'$), then 
  $|(N_{G'}[p_i]\cup N_{G'}[p_j])\cap D|\geq |\{x_i,x_j,p_i\}|=3$.
  If none of $p_i,p_j$ belongs to $N'$, then there must exists some 
  $q_i,q_j \in J'$ such that $q_i,q_j$ dominate $p_i,p_j$, respectively.
  Hence, $|(N_{G'}[p_i]\cup N_{G'}[p_j])\cap D|\geq |\{x_i,x_j,q_i,q_j\}|=4$.\\  
  {\bf Case b. $q_i,q_j \in J$:} If both $q_i,q_j \in J'$, then it is 
  trivial that $|(N_{G'}[q_i]\cup N_{G'}[q_j])\cap D|\geq 3$. Suppose one of 
  $q_i,q_j$ belongs to $J'$ (without loss of generality let us assume $q_i \in J'$).
  As every vertex in $G'$ is double dominated, $q_j$ must be dominated 
  by two vertices in $J'$ or by either some $q_k$ in $J'$ and some $p_l$ in $N'$.
  In either case we get $|(N_{G'}[q_i]\cup N_{G'}[q_j])\cap D|\geq 3$. A 
  similar argument works even if none of $q_i,q_j$ belong to $J'$.\\  
  {\bf Case c.} $p_i \in N$ and $q_j \in J$: If none of $p_i$ and $q_j$ 
  belong to $D$, then the argument is trivial as each one is
  dominated by at least two vertices in $D$. If both belong to $D$, then 
  $|(N_{G'}[p_i]\cup N_{G'}[q_j])\cap D|\geq |\{p_i,x_i,q_j\}|=3$.
  If $p_i \in D$ and $q_j \notin D$ (the other case is similar), 
  then $|(N_{G'}[p_i]\cup N_{G'}[q_j])\cap D|\geq 3$
  holds as $q_j$ is double dominated.

 Likewise, we can argue for other pair combinations too. Therefore, 
 every pair of distinct vertices in $V'$ is dominated by at least 3 vertices in $D$.
 
\end{enumerate}

 Therefore $D$ is an LDS in $G'$ and 
 $|D| = |N'| + |J'| + |S|\leq k + 3\ell + 3n$.
 
 \noindent {\bf Sufficiency:} Let $D \subseteq V'$ be an LDS of size at 
 most $k + 3\ell + 3n$. We prove that $G$ has a
 vertex cover of size at most $k$ with the aid of the following claims.
 \begin{enumerate}[label=(\roman*)]
  \item $S \subset D$.  
  \item Every segment in the embedding must contribute at least 3 vertices  
  to $D$ and hence $|J \cap D| \geq 3\ell$, where $\ell$ is the total number of 
  segments in the embedding.  
  \item If $p_i$ and $p_j$ correspond to end vertices of an edge $(v_i,v_j)$ in $G$,
  and if both $p_i,p_j$ are not in $D$, then there must be at least $3\ell' + 1$ vertices in 
  $D$ form the segment(s) representing the edge $(v_i,v_j)$, where $\ell'$ is the 
  number of segments representing the edge $(v_i,v_j)$ in the embedding.
 \end{enumerate}
 
 Claim (i) directly follows from the definition of liar's dominating set.
 Observe that we added points $x_i,y_i, z_i$ such that $p_i$
 is adjacent to $x_i$, $x_i$ is adjacent to $y_i$, and $y_i$ is adjacent
 to $z_i$ in $G'$, i.e., $\{(p_i,x_i),(x_i,y_i),(y_i,z_i)\}\subset E'$, for each $i$. Hence,
 $z_i$ and $y_i$ must be in $D$ due to the first condition of the liar's domination.
 Also, every connected component of $D$ in $G$ must contain at least three vertices due to the second
 condition of liar's domination. Hence, $x_i \in D$. Therefore, any liar's dominating
 set of $G'$ must contain $\{x_i,y_i,z_i\}, 1\leq i \leq n$, i.e., $S \subset D$.
 
 Claim (ii) follows from the fact that only consecutive points are 
 adjacent (in $G'$) on any segment in the embedding. Let $\eta$ be a segment in the embedding 
 having vertices $q_i,q_{i+1},q_{i+2}$, and $q_{i+3}$.
 On contrary, assume that $\eta$ has only two of its vertices in $D$. 
 Note that both $q_{i+1}$ and $q_{i+2}$ can not be in $D$ simultaneously.
 If both are present in $D$, then they do not satisfy the second condition 
 as $q_i$ and $q_{i+3}$ are not in $D$, i.e., 
 $|(N_{G'}[q_{i+1}]\cup N_{G'}[q_{i+2}]) \cap D| = |\{q_{i+1},q_{i+2}\}| =2$; contradiction to $D$ is an LDS.
 Without loss of generality we assume that $q_{i+2} \notin D$ 
 (the similar argument works even if $q_{i+1} \notin D$). 
 If $q_i$ and $q_{i+1}$ are in $D$, then $|(N_{G'}[q_{i+1}]\cup N_{G'}[q_{i+2}]) \cap D| = |\{q_i,q_{i+1}\}| =2$. If $q_i$ and $q_{i+3}$ are in $D$, then 
 $|(N_{G'}[q_{i+1}]\cup N_{G'}[q_{i+2}]) \cap D| = |\{q_i,q_{i+3}\}| =2$.
 If $q_{i+1}$ and $q_{i+3}$ are in $D$, then 
 $|(N_{G'}[q_{i+1}]\cup N_{G'}[q_{i+2}]) \cap D| = |\{q_{i+1},q_{i+3}\}| =2$.
 In either case we arrived at a contradiction.
 
 Claim (iii) follows from Claim (ii). Let $(v_i,v_j)$ be an edge in $G$ such that 
 $p_i$ and $p_j$ are not in $D$. By Claim (ii) every segment must contribute 
 at least three vertices to $D$. Hence, the number of vertices in $D$ from the segments 
 representing the edge $(v_i,v_j)$ is at least $3\ell'$. We argue that if both $p_i$ 
 and $p_j$ are not in $D$, then the number of vertices in $D$ from the segments 
 representing the edge $(v_i,v_j)$ is at least $3\ell' + 1$. Suppose that there are 
 exactly $3\ell'$ vertices in $D$ from the segments. That is, no segment representing 
 the edge $(v_i,v_j)$ contains more than three vertices in $D$. Let $p_i,q_1,q_2,\ldots,q_{4\ell'},p_j$ 
 be the vertices encountered while traversing the segments from $p_i$. 
 If $\ell'=1$, the argument can be proven as in the proof of Claim (ii).
 Assume $\ell'>1$.
 
 \begin{figure}[ht]
  \centering
  \includegraphics[scale=0.85]{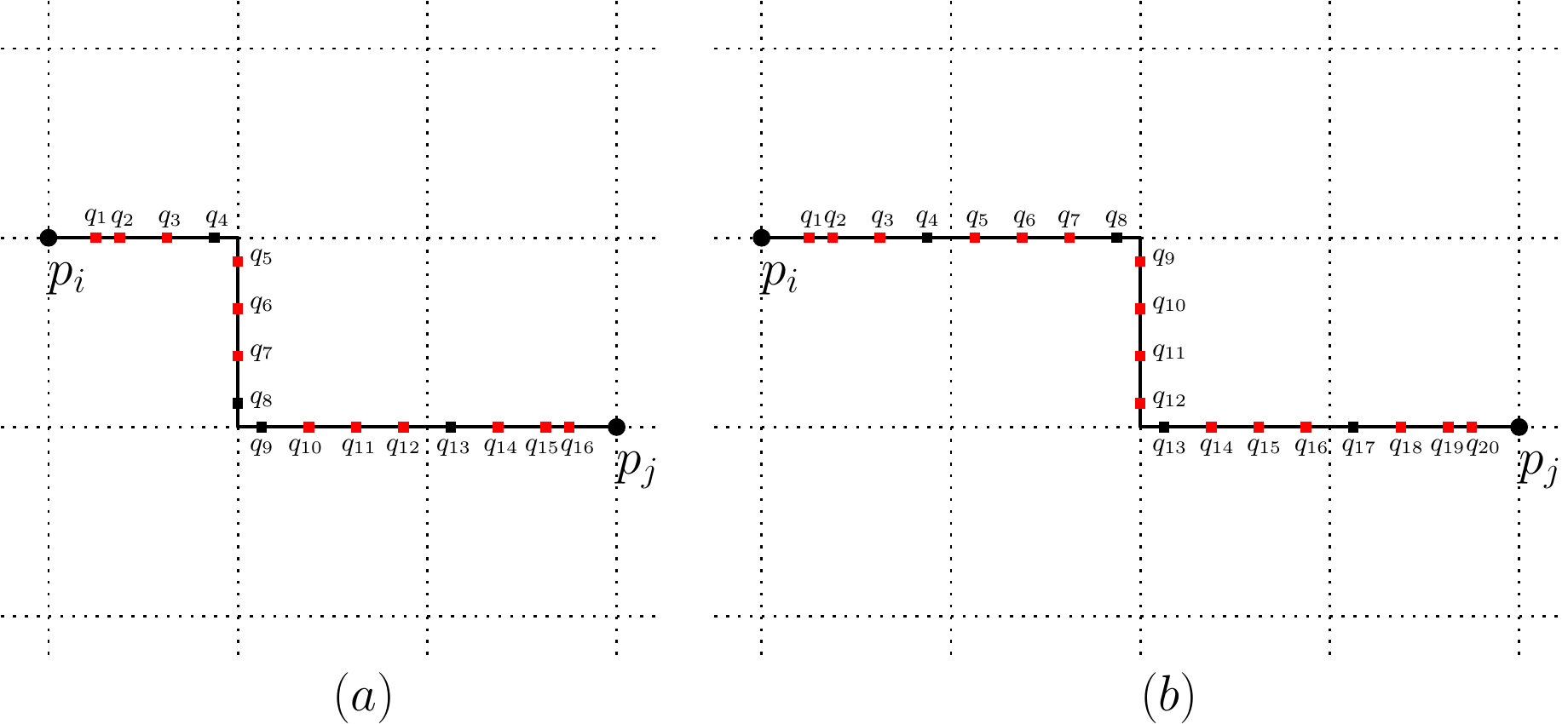}
  \caption{Illustration of Claim (iii). The vertices marked red must be in $D$.}\label{fig:claim3}
 \end{figure}
\begin{description}
 \item [Case a. $\ell'$ is even:] Since $p_i$ and $p_j$ are not in $D$ and 
 due to the second condition of the liar's domination, the vertices $q_1,q_2,q_3$ 
 from the first segment and $q_{4\ell'-2},q_{4\ell'-1},q_{4\ell'}$ from the 
 last segment must be in $D$. The vertices $q_4$ and $q_{4\ell'-3}$ can not 
 be in $D$ as we assumed that each segment contains exactly three vertices in $D$.
 If we continue in the same manner for the rest of the segments from both sides,
 we end up in not choosing the vertices $q_{2\ell'}$ and $q_{2\ell'+1}$ from the  
 $\frac{\ell'}{2}$-th and $(\frac{\ell'}{2}+1)$-th segments, respectively. 
 Note that $q_{2\ell'}$ is the last vertex on $\frac{\ell'}{2}$-th segment and 
 $q_{2\ell'+1}$ is the first vertex on $(\frac{\ell'}{2}+1)$-th segment and 
 $(q_{2\ell'},q_{2\ell'+1})$ is an edge in $G'$ (see Figure \ref{fig:claim3}(a)). Also, note that 
 $|(N_{G'}[q_{2\ell'}]\cup N_{G'}[q_{2\ell'+1}])\cap D|=|\{q_{2\ell'-1},q_{2\ell'+2}\}|=2$.
 Implies, the vertices $q_{2\ell'}$ and $q_{2\ell'+1}$ are not satisfying the second 
 condition, which is a contradiction to our assumption that $D$ is an LDS of $G'$.
 
 \item [Case b. $\ell'$ is odd:]
 If we proceed as in Case a, we can observe that $D$ must contain \emph{all} the four 
 vertices on $(\frac{\ell'}{2}+1)$-th segment, i.e., the middle segment, (see Figure \ref{fig:claim3}(b)). Which is a 
 contradiction to our assumption that no segment, representing the edge $(v_i,v_j)$, contains more than three vertices in $D$.
\end{description}

 We now shall show that, by removing and/or 
 replacing some vertices in $D$, a set of $k$ vertices from $N$ can be chosen 
 such that the corresponding vertices in $G$ is 
 a vertex cover. 
 The vertices in $S$ account for $3n$ vertices in $D$ (due to Claim (i)).
 Let $D = D \setminus S$ and $C = \{v_i \in V \mid p_i \in D \cap N\}$. 
 If any edge $(v_i,v_j)$ in $G$ has none of its end 
 vertices in $C$, then we do the following: consider the sequence of 
 segments representing the edge $(v_i,v_j)$ in the embedding. Since,
 both $p_i$ and $p_j$  are not in $D$, there must exist a segment having all its 
 vertices in $D$ (due to Claim (iii)). Consider the segment having its four 
 vertices in $D$. Delete any one of the vertices on the segment and introduce $p_i$ (or $p_j$).
 Update $C$ and repeat the process till every edge has at least one of its end vertices in $C$. 
 Due to Claim (ii), $C$ is a vertex cover in $G$ with $|C|\leq k$. 
 Therefore, \textsc{Lds-Udg} is NP-complete.
\end{proof}
\section{Approximation Algorithm}\label{appxalgo}
Banerjee and Bhore \cite{bhore} in their recent paper proposed an approximation algorithm and claim that their algorithm achieves a 5.5-factor approximation ratio for the MLDS problem in UDGs. However, their approximation analysis is erroneous. We first provide a counterexample defying their claim and then propose a simple 7.31-factor approximation algorithm for the said problem.

 For completeness here we give the idea of the algorithm proposed in \cite{bhore} briefly. As a first step, the point set $P$ (i.e., the set of disk centers) is sorted according to their $x$-coordinates. Now consider the left most point, say $p_i$, and consider $p_i$ in the solution. Next, compute the set of points of $P$ that are inside the circle centered at $p_i$ and of radius $\frac{1}{2}$, 1, and $\frac{3}{2}$. Let these sets be $Cov_{\frac{1}{2}}(C(p_i))$, $Cov_1(C(p_i))$, and $Cov_{\frac{3}{2}}(C(p_i))$, respectively. The points which lie outside the set $Cov_{\frac{3}{2}}(C(p_i))$, their corresponding disks of radius 1 do not contain any point from $Cov_{\frac{1}{2}}(C(p_i))$. So, it suffice to consider the points inside $Cov_{\frac{3}{2}}(C(p_i))$ to ensure liar's domination for $Cov_{\frac{1}{2}}(C(p_i))$. Since $p_i$ is the left most point in $P$, the set $Q = Cov_{\frac{3}{2}}(C(p_i))\setminus Cov_1(C(p_i))$ can contain at most five mutually independent points (i.e., the mutual distance between those five points is greater than one. In other words, the unit radius disks centered at those points do not contain the centers of other disks). In the next step (call it Case 1), for each point $q_i\in Q$, the algorithm chooses at most two points from the set $S(q_i)= Cov_{\frac{1}{2}}(C(p_i)) \cap Cov_1(C(q_i))$ in the solution, if available, where $Cov_1(C(q_i))$ is the set of points lying in the unit disk centered at $q_i$. After selection of these points, $Q$ is updated to $Q \setminus Cov_1(C(q_i))$  and proceed to next point in $Q$. Thus, the algorithm picks at most $5\times 2 + 1 = 11$ points in this iteration. Let $S = \bigcup\limits_{q_i \in Q}S(q_i)$. However, $S$ could be an empty set due to either $Q=\emptyset$ or $S(q_i) = \emptyset$ for each $q_i \in Q$ (call it Case 2). If $S=\emptyset$ or $|S| < 2$, then the algorithm chooses at most 4 points (including $p_i$) from $Cov_1(C(p_i))$ depending on the cardinality of $Cov_\frac{1}{2}(C(p_i))$. Thus, in this case the algorithm picks fewer than 11 points from $Cov_1(C(p_i))$. The points chosen so far ensures the liar's domination for the points in $Cov_\frac{1}{2}(C(p_i))$. Now, $P$ is updated to $P \setminus Cov_\frac{1}{2}(C(p_i))$, and the process is repeated (with the next leftmost point, say $p_j$) until $P$ is empty.

For each point $p_i\in P$, any optimal solution should contain at least two points from $Cov_1(C(p_i))$ due to the first condition of liar's domination, and the algorithm chooses at most 11 points.  Thus, the authors claim that the proposed algorithm is a $\frac{11}{2}$-factor approximation by the charging argument 11 points in the solution returned by algorithm can be charged to two points in the optimal solution. But, the same two points in the optimal solution could be charged multiple times.

Suppose $p_i$ and $p_j$ are the left most points considered in two successive iterations, respectively. There may be a case that the algorithm could end up by choosing a set of 11 points in the solution to dominate $Cov_{\frac{1}{2}}(C(p_j))$ for which the same optimal solution for $Cov_{\frac{1}{2}}(C(p_i))$ is enough to ensure liar's domination for $Cov_{\frac{1}{2}}(C(p_j))$. We elaborate our claim in detail with an example.

Consider the set of points in Figure \ref{fig:counter}(a) as an instance to the algorithm. The points are sorted according to their $x$-coordinates. Let the leftmost point be $p_i$ (see Figure \ref{fig:counter}(b)). The points $\{q_1,q_2,q_3,q_4,q_5\}\in Q$ and are five mutually independent points chosen by the algorithm such that $Cov_{\frac{1}{2}}(C(p_i))\cap Cov_1(C(q_j))=2$, for $j=1,2,\ldots,5$. Along with $p_i$, the total number of points chosen in this iteration is 11. Update $P=P\setminus Cov_{\frac{1}{2}}(C(p_i))$. In the next iteration, $p_j$ is the leftmost point (see Figure \ref{fig:counter}(c)) and $\{q_6,q_7,q_8,q_9,q_{10}\}\in Q$ are five mutually independent points chosen by the algorithm so that $Cov_{\frac{1}{2}}(C(p_j))\cap Cov_1(C(q_j))=2$, for $i=6,7,\ldots, 10$. The algorithm chooses 11 points (including $p_j$) in the solution. Observe that in both the iterations the algorithm doesn't enter Case 2 and, hence, chooses 22 points. In fact, any two (resp. three) red points (see  Figure \ref{fig:counter} (c)) are sufficient to ensure the liar's domination first (resp. second) condition for the point sets $Cov_{\frac{1}{2}}(C(p_i))$ and $Cov_{\frac{1}{2}}(C(p_j))$. After a few iterations $p_k$ will be chosen as the next left most point and the algorithm chooses 11 points (by Case 1) in the solution (see Figure \ref{fig:counter}(d)). However, the same three red points ensures liar's domination for $Cov_\frac{1}{2}(C(p_i))$, $Cov_\frac{1}{2}(C(p_j))$, and $Cov_\frac{1}{2}(C(p_k))$. So the approximation factor of the algorithm proposed in \cite{bhore} is at least 11.
 
\begin{figure}[!ht]
\begin{subfigure}[b]{.5\textwidth}  
  \centering
  \includegraphics[scale=.82]{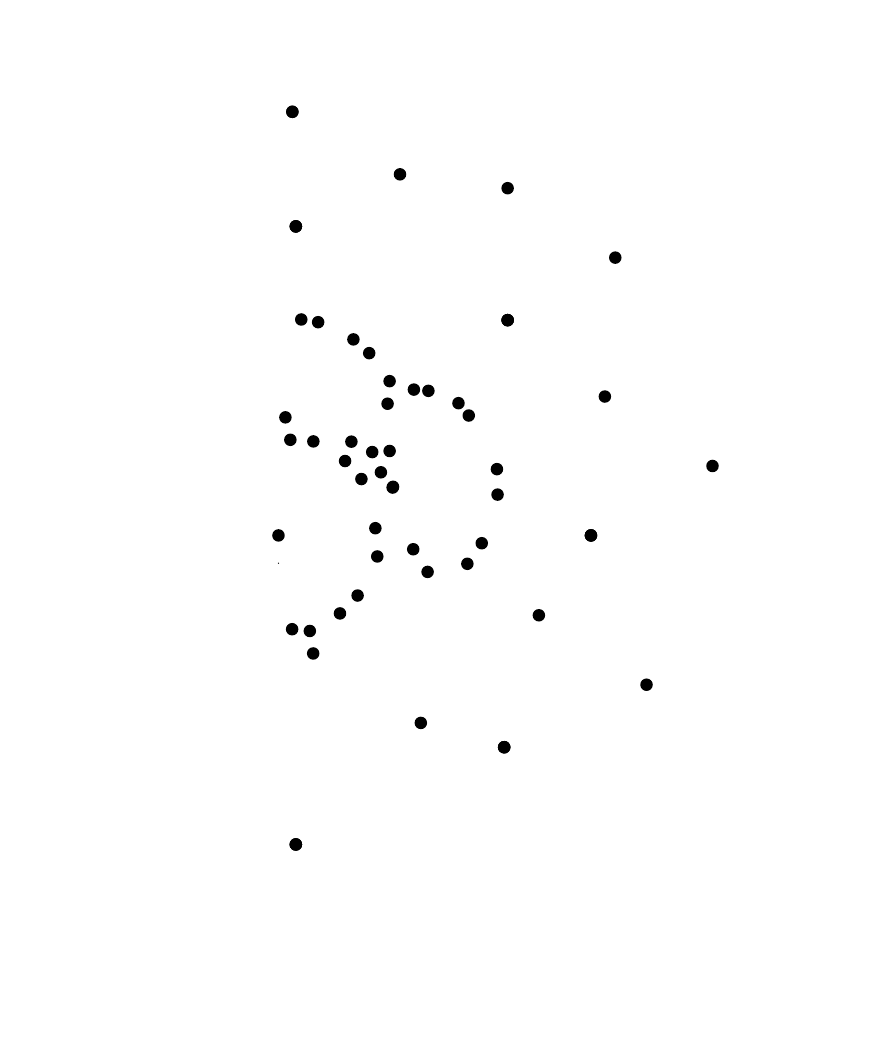}
  \caption{}  
\end{subfigure}%
\hfill
\begin{subfigure}[b]{.5\textwidth}  
  \centering
  \includegraphics[scale=.82]{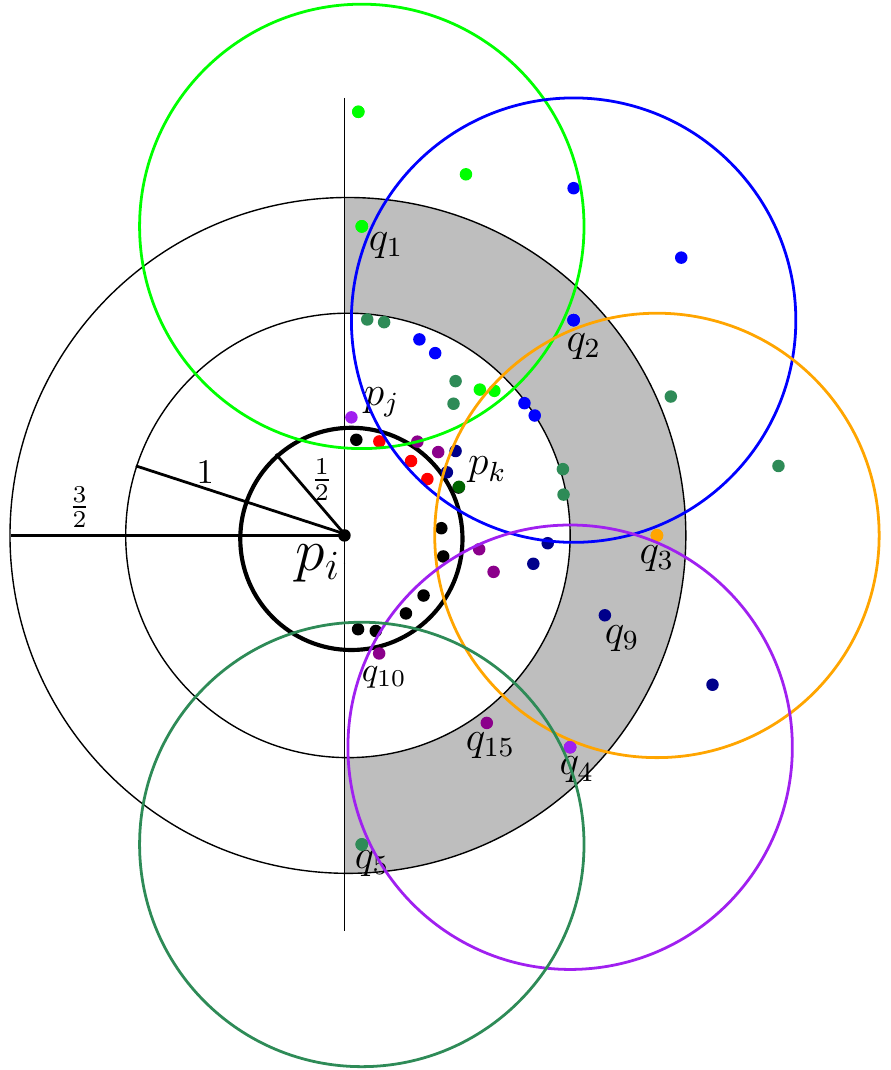}
  \caption{}  
\end{subfigure}%
\hfill
\begin{subfigure}[b]{.5\textwidth}  
  \centering
  \includegraphics[scale=.82]{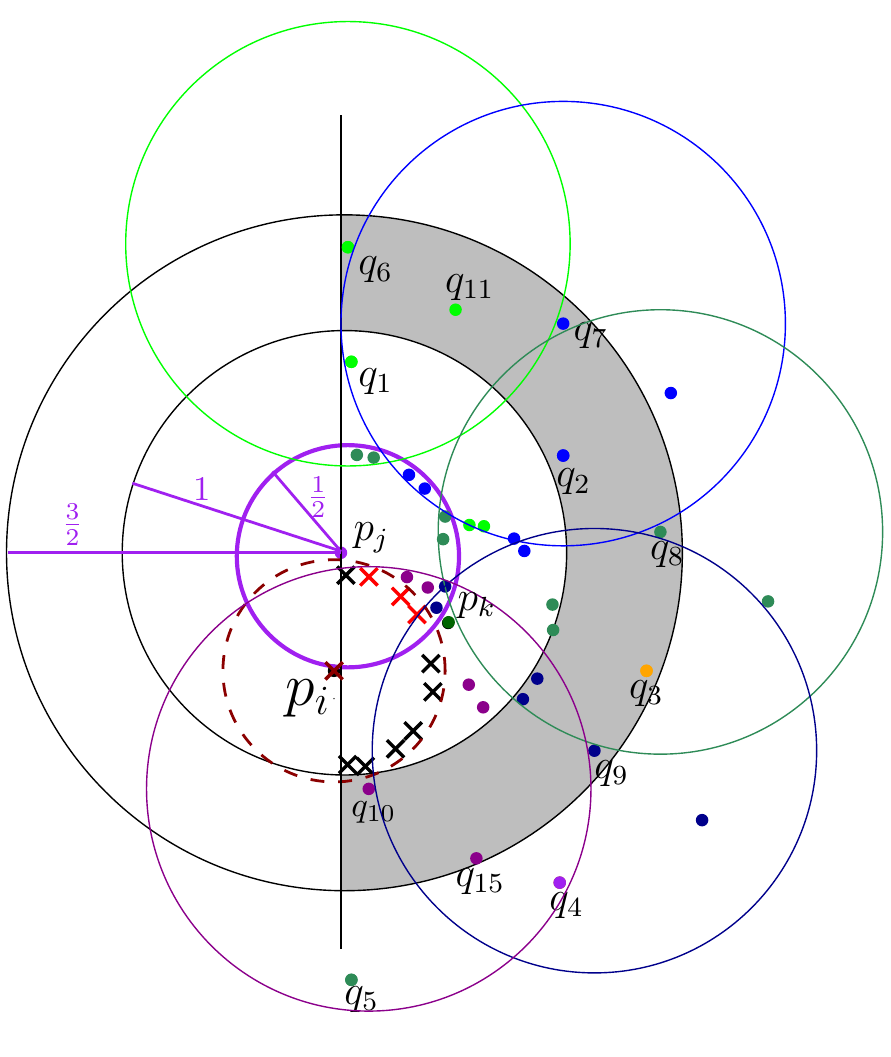}
  \caption{}  
\end{subfigure}
\hfill
\begin{subfigure}[b]{.5\textwidth}  
  \centering
  \includegraphics[scale=.82]{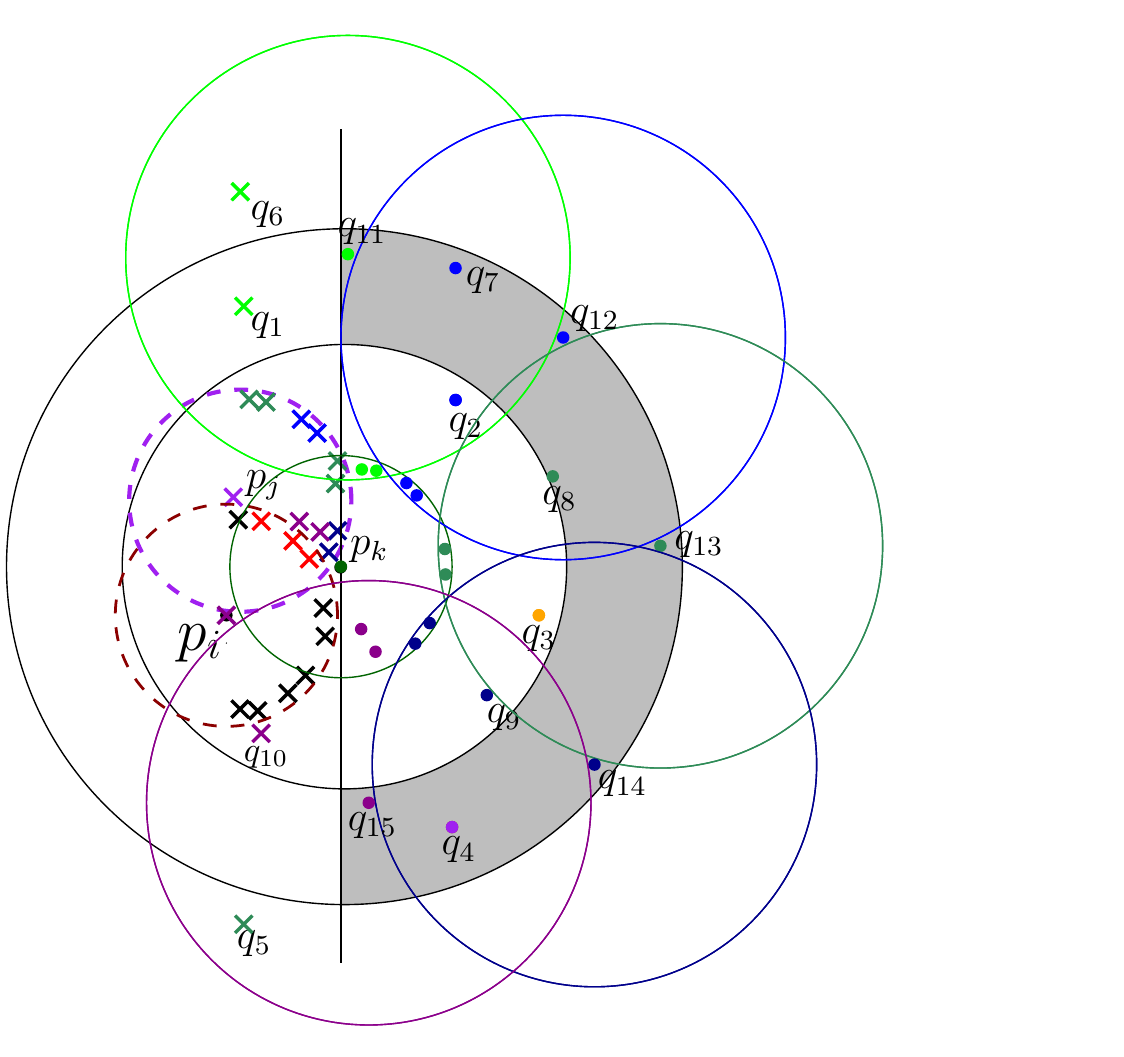}
  \caption{}  
\end{subfigure}
\caption{(a) A point set: an instance, (b) 11 points chosen from $Cov_\frac{1}{2}(C(p_i))$ out of which 3 red points are in optimal solution, and (c) and (d) the selected red points for $Cov_\frac{1}{2}(C(p_i))$ will ensure liar's domination for $Cov_\frac{1}{2}(C(p_j))$ and $Cov_\frac{1}{2}(C(p_k))$.}\label{fig:counter}  
\end{figure}     

\subsection{A 7.31-factor approximation algorithm}
In this Subsection, we propose a 7.31-factor approximation algorithm (see Algorithm \ref{algo:lds_apprx}) for minimum liar's dominating set (MLDS) problem in UDGs. The basic idea of the algorithm is: sequentially compute three maximal independent sets in the given UDG and add extra vertices, if necessary, to ensure liar's domination. In \cite{shang2008} Shang et. al. established a relation between maximal independent set and minimum $k$-dominating set\footnote{A minimum $k$-dominating set $D$ of $G$ is a minimum dominating set of $G$ with the property that every vertex not in $D$ should have at least $k$ dominators in $D$.} in UDGs. By using their result, we can have a 10-factor approximation algorithm for liar's dominating set in UDGs. In the following lemma, the proof idea is similar to \cite{shang2008}, we establish a relation between the cardinalities of maximal independent set and minimum liar's dominating set to obtain a 7.31-factor approximation algorithm for the MLDS problem in UDGs.

\begin{lemma}\label{lem:relation}
 Let $G=(V,E)$ be a UDG. If $I$ and $D_{opt}$ denote a maximal independent set and an MLDS of $G$, respectively, then $|I| \leq \sqrt{\frac{10}{3}}|D_{opt}|$.
\end{lemma}

\begin{proof}
 Let $I'=I\cap D_{opt}$, $X=I\setminus I'$, $Y=D_{opt}\setminus I'$. For $u,v\in X$, let $c_{u,v}$ denote the number of vertices in $Y$ which lie in the closed neighborhoods of $u$ and $v$ in $G$, i.e., $c_{u,v}=|(N[u]\cup N[v])\cap Y|$. As $D_{opt}$ is a liar's dominating set of $G$, $c_{u,v}\geq 3$ for each $u,v\in X$, and  we get $\sum_{u,v\in X}c_{u,v}\geq 3\cdot \frac{|X|(|X|-1)}{2}$. For  $u',v'\in Y$, analogues to $c_{u,v}$, let $d_{u',v'}=|(N[u']\cup N[v'])\cap X|$. As $G$ is a UDG, for each vertex in $Y$ there can be at most 5 independent vertices in its neighborhood, and thus $d_{u',v'}\leq 10$ for each $u'$ and $v'$ in $Y$. Hence, we get $10\cdot \frac{|Y|(|Y|-1)}{2} \geq \sum_{u',v'\in Y}d_{u',v'}$. Note that the number of edges in $E$ induced between $X$ and $Y$ is $\sum_{u,v\in X}c_{u,v}(=\sum_{u',v'\in Y}d_{u',v'})$. Thus, we have $3\cdot \frac{|X|(|X|-1)}{2}\leq 10\cdot \frac{|Y|(|Y|-1)}{2}$, which implies 
 $|X| \leq \sqrt{\frac{10}{3}}|Y|$. Therefore, $|I| = |X| + |I'| \leq \sqrt{\frac{10}{3}}|Y| + |I'| \leq \sqrt{\frac{10}{3}}|D_{opt}|$.
\end{proof}
\begin{algorithm}[!ht]
 \caption{Liar's dominating set in UDG}\label{algo:lds_apprx}
\begin{algorithmic}[1]
\Require An UDG $G=(V,E)$
\Ensure A liar's dominating set $D$ of $G=(V,E)$
 \State $i\leftarrow 0$, $I_i \leftarrow \emptyset$, and $D \leftarrow \emptyset$
      \For {$(i=1 \;to\; 3)$}\label{loop:first_for_start}
	  \If {($V\neq \emptyset$)}
	      \State $I_i\leftarrow MIS(V)$ \Comment{$MIS(\cdot)$ returns a maximal independent set} \label{mis:subroutine}
	      \State $D \leftarrow D \cup I_i$; $V\leftarrow V\setminus I_i$
	  \EndIf
      \EndFor \label{loop:first_for_end}
\For {every $u \in I_1$} \label{loop:sec_for_start}
      \If{$N_G(u)\cap (I_2\cup I_3)=\emptyset$}
      \State let $v \in N_G(u)$
	  \State $D = D \cup \{v\}$\label{line:add_v}
	  \ElsIf{$|N_G(u)\cap (I_2\cup I_3)|=1$}
	  \State let $w$ be a neighbor of $v\in N_G(u)\cap (I_2\cup I_3)$ such that $w\neq u$
	  \State $D = D \cup \{w\}$ \label{line:add_w}
      \EndIf
 \EndFor \label{loop:sec_for_end}
 \State \bf{return} $D$
 
\end{algorithmic}
\end{algorithm}

\begin{lemma}\label{lem:lds}
 The set $D$ returned by Algorithm \ref{algo:lds_apprx} is an LDS of $G$.
\end{lemma}
\begin{proof}
 Algorithm \ref{algo:lds_apprx} sub-sequentially computes three maximal independent sets $I_1$, $I_2$, and $I_3$ in $G$ (see line numbers \ref{loop:first_for_start}-\ref{loop:first_for_end}). Let $I = I_1 \cup I_2 \cup I_3$.
 Note that any vertex not in $I$ has a neighbor (dominator) in each $I_1, I_2$, and $I_3$. Thus, each vertex (resp. every pair of distinct vertices) not in $I_1 \cup I_2 \cup I_3$ satisfies the first (resp. second) condition of liar's domination. Also, every vertex in $I_3$ is adjacent to a vertex in $I_1$ and $I_2$. Thus, the vertices in $I_3$ satisfy both the conditions of the lair's domination.  Similarly, every vertex in $I_2$ is adjacent to a vertex in $I_1$ (otherwise, $I_1$ cannot be a maximal independent set) and, hence, the vertices in $I_2$ satisfy both the conditions of the lair's domination.  For any vertex $u \in I_1$, if $N_G(u)\cap (I_2\cup I_3)=\emptyset$, then the algorithm adds an arbitrary neighbor $v$ of $u$ to $D$ (see line number \ref{line:add_v}).
 If $N_G(u)\cap (I_2\cup I_3)\neq \emptyset$, then $u$ has neighbor in $I_2\cup I_3$. In either case the vertices in $I_1$ satisfy the two conditions of the liar's domination. 
 
 For any pair of distinct vertices $u \in I$ and $v \in V \setminus I$, the second condition is already satisfied as $v$ has a neighbor in each $I_1, I_2$, and $I_3$. Similarly, for any pair of distinct vertices $u \in I_1$, and $v \in I_2$, if $v$ has multiple neighbors in $I_1$ or $u$ has multiple neighbors in $I_2$, then the second condition is satisfied. If $u$ is the only neighbor of $v$ in $I_1$ and vice versa, then the algorithm adds an arbitrary neighbor $w$ of $v$ to $D$, see line number \ref{line:add_w}, and thus the second condition is ensured for $u$ and  $v$.  If $v \in I_3$, the second condition is trivially holds as $v$ has a neighbor in each $I_1$ and $I_2$. Therefore, $D$ is an LDS in $G$.
\end{proof}

\begin{theorem}
 For a given UDG $G =(V,E)$, Algorithm \ref{algo:lds_apprx}
 achieves approximation ratio 7.31 for the MLDS problem in $O(|V|+|E|)$ time. 
\end{theorem}
\begin{proof}
Let $D^*$ be an MLDS of $G$. Algorithm \ref{algo:lds_apprx} sequentially computes three maximal independent sets $I_1$, $I_2$, and $I_3$ in $G$ and $I = I_1 \cup I_2 \cup I_3$ is not necessarily be an LDS of $G$ as there might be some vertices with one of the following cases: (i) a vertex $u \in I_1$ not satisfying the first condition, or (ii) a pair of distinct vertices $u \in I_1$ and $v\in I_2$ not satisfying the second condition of the liar's domination. In the former case we add an arbitrary neighbor of $u$, and in the latter case we add an neighbor $w$ of $v$. Note that in either case such a neighbor is guaranteed to exist in $G$, and $|D| \leq |I| + |I_1|$. Without loss of generality we can assume that $|I_3| \leq |I_2| \leq |I_1|$. Therefore, $|D| \leq 4|I_1|\leq 4\cdot \sqrt{\frac{10}{3}} |D^*|\leq 7.31|D^*|$ (by Lemma \ref{lem:relation}). The running time follows as Algorithm \ref{algo:lds_apprx} uses the subroutine $MIS(\cdot)$ 
(in line number \ref{mis:subroutine}) to compute a maximal independent set. 
\end{proof}

\section{Approximation Scheme} \label{ptas}
In this section, we propose a PTAS for the MLDS problem in UDGs, 
i.e., for a given UDG $G=(V,E)$ and a parameter $\epsilon > 0$, we propose an
algorithm which produces a liar's dominating set of size no
more than $(1 + \epsilon)$ times the size of a minimum liar's dominating set in G.
We use $\delta_G(u,v)$ to denote the number of edges on a shortest path 
between $u$ and $v$ in $G$. For $A,B \subseteq V$, $\delta_G(A,B)$ denotes the
distance between $A$ and $B$ and is defined as
$\delta_G(A,B)=\min_{u\in A,v\in B}\{\delta_G(u,v)\}$. 
For $A\subseteq V$, $D(A)$ and $D_{opt}(A)$ denote an LDS and an optimal (minimum size) 
LDS of $A$ in $G$, respectively. We define the closed 
neighborhood of a set $A\subseteq V$ as $N_G[A] = \bigcup_{v\in A}N_G[v]$.

\begin{figure}[h]
\centering
 \includegraphics[scale=0.55]{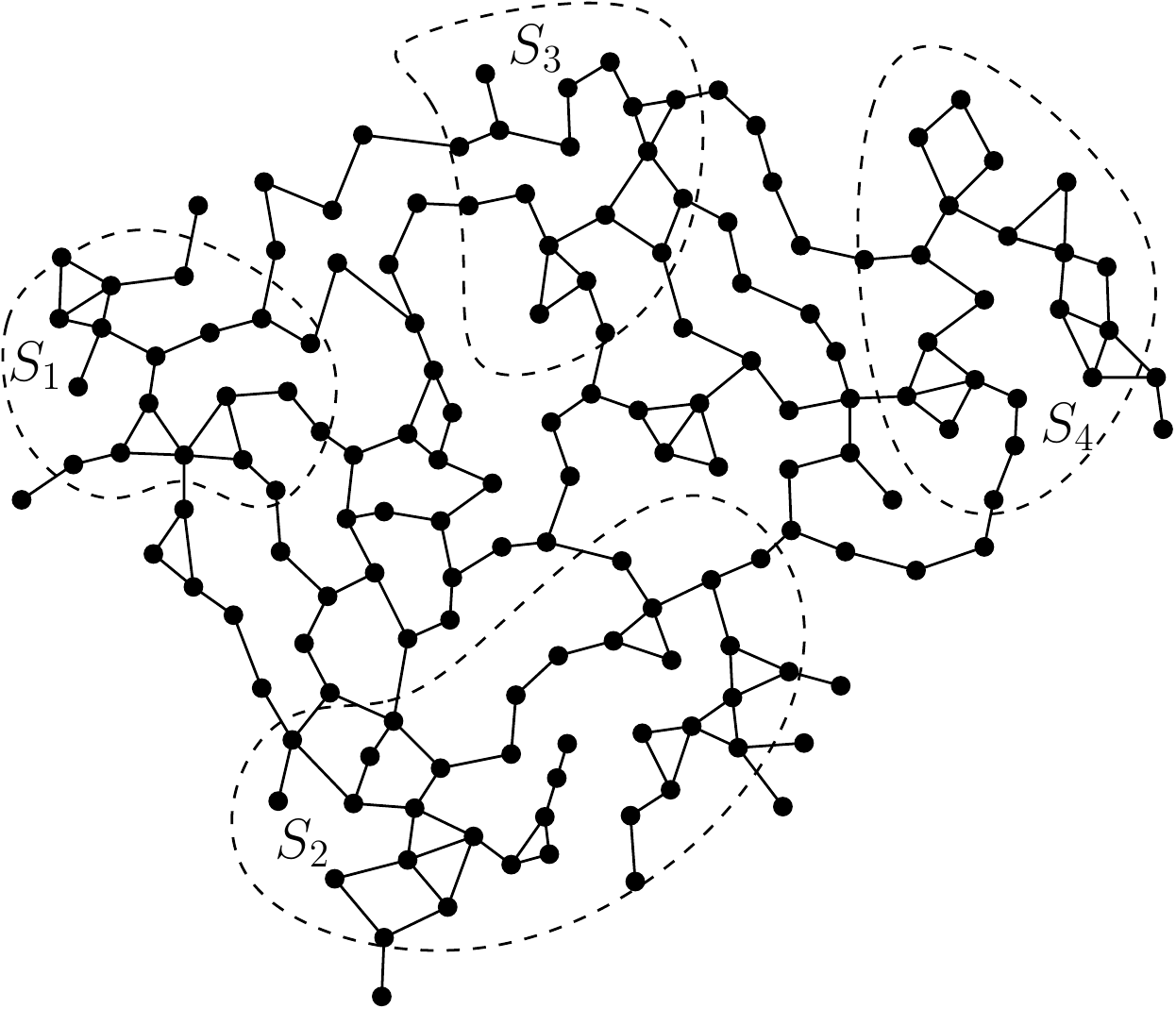}
 \caption{A 4-separated collection $S=\{S_1,S_2,S_3,S_4\}$}\label{sc}
\end{figure}

The proposed PTAS is based on the concept of $m$-separated collection
of subsets of $V$ ($m\geq 4$). Let $G=(V,E)$ be a UDG. A collection $S = \{S_1,S_2,\ldots,S_k\}$ such 
that $S_i \subseteq V$ for $i=1,2,\ldots,k$, is said to be an $m$-separated
collection, if $\delta_G(S_i,S_j)>m$, for $1\leq i\leq k$ and $1 \leq j \leq k$ (see Figure \ref{sc} 
for a 4-separated collection). Nieberg and Hurink \cite{nieberg} introduced 2-separated collection to
propose a PTAS for the minimum dominating set problem in UDGs and our PTAS follows form it. However, the algorithm in \cite{nieberg} cannot be directly applied as in intermediate steps of the algorithm we need to add extra nodes (see line numbers \ref{algo:if_start}-\ref{algo:if_end} in Algorithm \ref{lds}) to ensure the liar's domination. We argue that the extra nodes added are small enough and do not effect the approximation factor.

\begin{lemma}\label{4-sc}
 Let $S=\{S_1,S_2,\ldots,S_k\}$ be an $m$-separated collection. If 
 $|S_i|\geq 3$ for $1 \leq i \leq k $, then  
 $\sum_{i=1}^{k}|D_{opt}(S_i)| \leq |D_{opt}(V)|$.
\end{lemma}

\begin{proof}

 Observe that $N_G[S_i]\cap N_G[S_j] = \emptyset$  for $i \neq j$ and $1\leq i,j\leq k$. Also, 
 $D_{opt}(S_i)\cap D_{opt}(S_j) = \emptyset$  as $S_i$ and $S_j$ 
 are $m$-separated. 
 Let $S_i'=\{u\in V\mid v\in S_i \text{ and } \delta_G(u,v)\leq 2\}$ for $i=1,2,\ldots,k$. 
 Observe that $S_i\subseteq S_i'$ and $S_i'\cap D_{opt}(V)$ is a liar's 
 dominating set of $S_i$ for $i=1,2,\ldots,k$. Since, $\delta_G(S_i,S_j)>m (\geq 4)$ 
 for $i\neq j$, implies $S_i'\cap S_j'=\emptyset$. Therefore, 
 $(S_i' \cap D_{opt}(V))\cap (S_j' \cap D_{opt}(V)) = \emptyset$ and 
 $\bigcup_{i=1}^k(S_i'\cap D_{opt}(V))\subseteq D_{opt}(V)$. Also, 
 $|D_{ opt }(S_i)| \leq |S_i'\cap D_{opt}(V)|$ for $i=1,2,\ldots,k$ as 
 $S_i' \cap D_{opt}(V)$ is a liar's dominating set of $S_i$, and 
 $D_{opt}(S_i)$ is a minimum size liar's dominating set. Thus, 
 $\sum_{i=1}^k |D_{opt}(S_i)|\leq\sum_{i=1}^k|S_i'\cap D_{opt}(V)|\leq |D_{opt}(V)|$.
\end{proof}

\begin{lemma}\label{lem_ptas}
 Let $S=\{S_1,S_2,\ldots,S_k\}$ be an $m$-separated collection,
 and $N_1,N_2,\ldots,N_k$ be subsets of $V$ with $S_i\subseteq N_i$ for 
 all $i=1,2,\ldots,k$. If there exists $\rho\geq 1$ such that 
 $|D_{opt}(N_i)| \leq \rho |D_{opt}(S_i)|$ holds for all $i=1,2,\ldots,k$, and if
 $\bigcup_{i=1}^k D_{opt}(N_i)$ is a liar's dominating 
 set in $G$, then the value of  $\sum_{i=1}^k |D_{opt}(N_i)|$ 
 is at most $\rho$ times the size of a minimum liar's dominating set in $G$.
\end{lemma}
\begin{proof}
$\sum_{i=1}^k |D_{opt}(N_i)|\leq \rho \sum_{i=1}^{k}|D_{opt}(S_i)|\leq \rho |D_{opt}(V)|$. 
The latter inequality follows from Lemma \ref{4-sc}. 
\end{proof}

\subsection{Algorithm}
In this section, we discuss the construction of a 4-separated collection 
$S = \{S_1,S_2,\ldots,S_k\}$ and subsets $N_1,N_2,\ldots,N_k$ of $V$ such that $S_i\subseteq N_i$
for all $i=1,2,\ldots,k$. The algorithm proceeds in an iterative
manner. Initially $V_1=V$. In the $i$-th iteration the algorithm computes $S_i$ and $N_i$.
For a given $\epsilon > 0$, the $i$-th iteration of the algorithm 
starts with an arbitrary vertex $v\in V_i$ and increases the value of 
$r (=2,3,\ldots)$ as long as $|D(N_G^{r+4}[v])|>\rho|D(N_G^r[v])|$ holds.
Here, $D(N_G^{r+4}[v])$ and $D(N_G^r[v])$ are liar's dominating sets
of $N_G^{r+4}[v]$ and $N_G^r[v]$, respectively, and $\rho = 1+\epsilon$. The
smallest $r$ violating the above condition, say $\hat{r}$, is obtained.
Set $S_i = N_G^{\hat{r}}[v]$ and $U_i =N_G^{\hat{r}+4}[v]$. Now, the removal of 
$U_i$ from $V_i$ may lead to some isolated (i) vertex $u\in V_i$, and/or  
(ii) connected component with two vertices $u,w\in V_i$. In case (i), 
for each such vertex $u$ find $x,y\in U_i$ such that $\{u,x,y\}$ forms 
a connected component and update $U_i$ as follows: $U_i=U_i\setminus \{x,y\}$. 
In case (ii), for each such pair of vertices $u,w$ find $x\in U_i$ 
such that $\{u,w,x\}$ forms a connected component and update $U_i$ as 
follows: $U_i=U_i\setminus \{x\}$. Set $N_i = U_i$ and $V_{i+1}=V_i\setminus N_i$.
The process stops if $V_{i+1}=\emptyset$ and returns the sets $S_i$s and $N_i$s. 
The collection of the sets $S_i$s is a 4-separated collection.
The pseudo code is given in Algorithm \ref{lds}.

The liar's dominating set of a $r$-th neighborhood of a vertex $v$, $D(N_G^r[v])$,
can be computed with respect to $G$ as described in Algorithm \ref{algo:lds_apprx}.
Algorithm \ref{algo:lds_apprx} successively finds 
maximal independent sets $I_1,I_2$, and $I_3$.
Now, $I_1\cup I_2 \cup I_3$ is a liar's dominating set
for $N_G^r[v] \setminus I_1$ as every vertex not in $I_1$ either belongs
to $I_2\cup I_3$ or is adjacent to at
least one vertex in each $I_1,I_2$, and $I_3$. To ensure
the liar's domination for the vertices in $I_1$,
for each vertex $u$ in $I_1$ we add at most a vertex (see line numbers \ref{loop:sec_for_start}-\ref{loop:sec_for_end} 
in Algorithm \ref{algo:lds_apprx}).
In summary, Algorithm \ref{lds} deals with obtaining an 
$m$-separated collection $S=\{S_1,S_2,\ldots,S_k\}$ and collection $N=\{N_1,N_2,\ldots,N_k\}$ 
such that $S_i\subseteq N_i \subseteq V$ and using Algorithm \ref{algo:lds_apprx} 
(that deals with obtaining a liar's  dominating set of the
$r$-th neighborhood of a vertex) as a sub-routine, Algorithm \ref{lds} computes a liar's dominating set for $G$.

\begin{algorithm}[t]
\caption{Liar's dominating set}\label{lds}
\begin{algorithmic}[1]
\Require A unit disk graph $G=(V,E)$ with $|V|\geq 3$ and an arbitrary small $\epsilon>0$
\Ensure A liar's dominating set ${\cal D}$ of $V$
 \State $i\leftarrow 0$ and $V_{i+1}\leftarrow V$
 \State ${\cal D}\leftarrow \emptyset$ and $\rho \leftarrow  1 + \epsilon$
 \While{($V_{i+1}\neq \emptyset$)}
      \State pick an arbitrary $v\in V_{i+1}$
      \State $N^0[v]\leftarrow v$ and $r\leftarrow 2$
     
      \While{$|(D(N_G^{r+4}[v])|>\rho|D(N_G^r[v])|$} \label{condition}\Comment{call Algorithm \ref{algo:lds_apprx}}
	  \State $r\leftarrow r+1$
      \EndWhile
      \State $\hat{r}\leftarrow r$ \Comment{the smallest $r$ violating while condition in step \ref{condition}}
      \State $i \leftarrow i+1$ \Comment{the index $i$ keeps track of the number of iterations}
      
      \State $S_i \leftarrow N_G^{\hat{r}}[v]$ and $U_i \leftarrow N_G^{\hat{r}+4}[v]$
      \If{($V_{i+1}\setminus N_G^{\hat{r}+4}[v]$ contains isolated components of size 1 and/or 2)}\label{algo:if_start}
	 \For{(each component, $\{u\}$, of size 1)}
	    \State find $x,y\in U_i$ such that $\{u,x,y\}$ is a connected component
	    \State $U_i \leftarrow U_i\setminus\{x,y\}$
	 \EndFor
	 \For{(each component, $\{u,w\}$, of size 2)}
	    \State find $x\in U_i$ such that $\{u,w,x\}$ is a connected component
	    \State $U_i \leftarrow U_i\setminus\{x\}$
	 \EndFor
      \EndIf \label{algo:if_end}
      \State $N_i \leftarrow U_i$
      \State ${\cal D} \leftarrow {\cal D} \cup D(N_i)$ \Comment{call Algorithm \ref{algo:lds_apprx}}
      \State $V_{i+1}\leftarrow V_i \setminus N_i$
 \EndWhile
 \State \bf{return} ${\cal D}$
\end{algorithmic}
\end{algorithm}

\begin{lemma}\label{bound}
 $D(N_G^r[v])$ is an LDS of $N_G^r[v]$ in $G$ and $|D(N_G^r[v])|\leq O(r^2)$.
\end{lemma}

\begin{proof}
 Algorithm \ref{algo:lds_apprx} computes $D(N_G^r[v])$ by first computing
 maximal independent sets $I_1, I_2$, and $I_3$ subsequently and then it adds at most one vertex for
 each vertex in $I_1$ to ensure that $D(N_G^r[v])$ is a feasible
 solution. We can show $D(N_G^r[v])$ is an LDS of $N_G^r[v]$ in $G$ as in the proof of Lemma \ref{lem:lds}.
 Hence, $|D(N_G^r[v])|\leq 4\cdot |I_1|\leq 4\cdot \frac{\pi(r+1)^2}{\pi(1)^2}=O(r^2)$. 
 The latter inequality follows from the standard area argument, 
 the number of non-intersecting unit disks can be packed in 
 a larger disk of radius $r+1$ centered at $v$.
\end{proof}

\begin{lemma}\label{violate}
 In each iteration of Algorithm \ref{lds}, there exists an $r$ violating the condition
 $|(D(N_G^{r+4}[v])|>\rho|D(N_G^r[v])|$, where $\rho = 1 + \epsilon$.
\end{lemma}
\begin{proof}
We prove the lemma by contradiction. Suppose there exists $v\in V$
such that $|(D(N_G^{r+4}[v])|>\rho|D(N_G^r[v])|$ for $r=2,3,\ldots$.
Observe that $|D(N_G^2[v])|\geq 3$ as there exists at least three vertices in $G$.\\
If $r = 4k$, $4(r+1)^2\geq|(D(N_G^r[v])|>\rho|D(N_G^{r-4}[v])|>\cdots>\rho^\frac{r}{4}|D(N_G^2[v])|\geq3\rho^\frac{r}{4}$, and \\ if $r =4k + s$ for $1 \leq s \leq 3$, \\
$4(r+1)^2\geq|(D(N_G^r[v])|>\rho|D(N_G^{r-4}[v])|>\cdots>\rho^\frac{r-1}{4}|D(N_G^3[v])|\geq3\rho^\frac{r-1}{4}$.

In both the cases the first inequality follows from Lemma \ref{bound}.
 Hence,
\begin{equation}\label{eqn}
  4(r+1)^8 >
  \begin{cases}
    {\rho}^r, & \text{if $r$ is $4k$}\\
    {\rho}^{r-1}, & \text{if $r$ is $4k+s$}
  \end{cases}
\end{equation}

The right hand part in inequality (\ref{eqn}) is an exponential function in $r$ and the left hand
part is a polynomial in $r$, for arbitrarily large $r$ none of the inequalities can be true. Hence we 
arrived at contradiction. Thus there exists an $r$ violating the condition.
\end{proof}
The following lemma suggests that the smallest $r$ violating inequality (\ref{eqn}) is bounded by a 
constant that depends only on $\epsilon$.
\begin{lemma}\label{boundonr}
 The smallest $r$ violating the inequality (\ref{eqn}) is bounded by $O(\frac{1}{\epsilon}\log \frac{1}{\epsilon})$.
\end{lemma}
\begin{proof}
Let $\hat{r}$ be the smallest $r$ violating the
inequalities in (\ref{eqn}). Using the inequalities (i) 
$\log(1+\epsilon)>\frac{\epsilon}{2}$ for $0< \epsilon < 1$, (ii) $\log x < x$ for $x>1$, 
and (iii) $\log \frac{1}{\epsilon} \geq 1$ for $\epsilon \leq \frac{1}{10}$, we show 
$\hat{r}\leq O(\frac{1}{\epsilon}\log \frac{1}{\epsilon})$. Let 
$x=\frac{c}{\epsilon}\log \frac{1}{\epsilon}$. Consider the inequality 
$4(x+1)^8\leq 4(2x)^8\leq (1 + \epsilon)^x$. The former inequality trivially holds for the choice of $x$ 
and for any $\epsilon>0$, and taking the logarithm on both sides of the latter inequality, we 
get $\frac{\log 4 + 8\log 2x}{x}\leq \log(1+\epsilon)$. By inequality (i), now, it suffice to show 
that $\frac{\log 4+8\log 2x}{x}\leq \frac{\epsilon}{2}$. Using the inequalities (ii) and (iii), the choice 
of $x$ satisfies the inequality for any constant $c$ satisfying $\log 2c + 2 < \frac{c}{16}$. 
\end{proof}

\begin{lemma}\label{rtime}
 For a given $v\in V$, liar's dominating set $D_{opt}(N_i)$ of $N_i$ can be computed 
 in polynomial time.
\end{lemma}

\begin{proof}
 Note that $N_i\subseteq N_G^{r+4}[v]$. The size of a liar's dominating set 
 $D(N_i)$ of $N_i$ is bounded by $O(r^2)$ (by Lemma \ref{bound}). Again, 
 $r=O(\frac{1}{\epsilon}\log \frac{1}{\epsilon})$ by Lemma \ref{boundonr}. 
 Therefore, the size of the minimum size liar's dominating set $D_{opt}(N_i)$ 
 of $N_i$ is bounded by a constant. The process of checking whether a given 
 set is a liar's dominating set or not can be done in polynomial-time. 
 Therefore, we can consider every subset of $N_i$ as a possible liar's 
 dominating set and check whether it is a liar's dominating set or not 
 in polynomial-time. Finally, the minimum size liar's dominating set is 
 reported. 
\end{proof}

\begin{lemma}
 For the collection of neighborhoods $\{N_1,N_2,\ldots,N_k\}$
 created by Algorithm \ref{lds}, the union ${\cal D}=\bigcup_{i=1}^k D(N_i)$
 is a liar's dominating set in $G$.
\end{lemma}

\begin{proof}
 We first prove that for every $v\in V, |N_G[v]\cap{\cal D}|\geq 2$. Observe that
 $\bigcup_{i=1}^k N_i=V$ as $V_{i+1}=V_i\setminus N_i$
 and $N_i \subseteq V_i$. Thus, every vertex $v \in N_i$ for some
 $1\leq i \leq k$. By Lemma \ref{bound}, $|N_G[v]\cap{\cal D}|\geq 2$ is satisfied.

 Now we prove the second condition. Let $u,v \in V$ be any two arbitrary vertices. 
 The following cases may arise.\\
  {\bf Case 1.} $u,v\in N_i$ for some $1\leq i \leq k$\\
  Since $D(N_i)$ is the liar's dominating set of $N_i$ in $G$, we
  have, $|(N_G[u]\cup N_G[v])\cap D(N_i)|\geq 3$ for every
  $u,v\in N_i$. Hence, $|(N_G[u]\cup N_G[v])\cap {\cal D}|\geq 3$
  for every $u,v\in V$.\\
  {\bf Case 2.} $u\in N_i$ and $v\in N_j$ for some $i\neq j$ and $1\leq i,j\leq k$\\   
  If $u$ and $v$ are not adjacent in $G$, the proof is trivial. Hence,
  we assume that $(u,v)\in E$ i.e., $u$ and $v$ are adjacent
  in $G$. Now the following sub-cases may arise.\\  
   (a) $u\in D(N_i)$ and $v\in D(N_j)$\\
    Observe that $|N_G[u]\cap D(N_i)|\geq 2$ and $|N_G[v]\cap D(N_j)|\geq 2$
    as $D(N_i)$ and $D(N_j)$ are liar's dominating sets of $N_i$ and $N_j$,
    respectively. Hence, $u$ has a neighbor, say $w$, in $D(N_i)$,
    similarly $v$ has also a neighbor, say $x$, in $D(N_j)$. However,
    maybe $w=x$ or maybe not. In either case
    $|(N_G[u]\cup N_G[v])\cap {\cal D}|\geq 3$ holds.\\
  (b) $u \notin D(N_i)$ and $v\in D(N_j)$ (the other case proof is similar)\\
   Since $D(N_i)$ is a liar's dominating set of $N_i$, we have 
   $|N_G[u]\cap D(N_i)|\geq 2$. Hence, $|(N_G[u]\cup N_G[v])\cap {\cal D}|\geq 3$ 
   is true as $v$ is part of the solution.\\
  (c) $u\notin D(N_i)$ and $v\notin D(N_j)$ (The proof is similar to the previous cases).   
\end{proof}

\begin{corollary}\label{cor_main}
 For the collection $N=\{N_1,N_2,\ldots,N_k\}$ created by Algorithm \ref{lds}, 
 the union ${\cal D}^*= \bigcup_{i=1}^kD_{opt}(N_i)$ is a liar's 
 dominating set.
\end{corollary}

\begin{theorem}
 For a given UDG, $G=(V,E)$, and an $\epsilon > 0$, 
 we can design a $(1+\epsilon)$-factor approximation algorithm 
 to find an LDS in $G$ with running time $n^{O(c^2)}$, where 
 $c=O(\frac{1}{\epsilon}\log \frac{1}{\epsilon})$.
\end{theorem}

\begin{proof}
Note that Algorithm \ref{lds} generates the collection of sets 
$S=\{S_1,S_2,\ldots,S_k\}$ and $N=\{N_1,N_2,\ldots,N_k\}$ such that 
$S$ is a 4-separated collection of $V$ with $S_i\subseteq N_i$ 
for each $i\in\{1,2,\ldots,k\}$ and  $\bigcup_{i=1}^k N_i=V$ 
with $N_i\cap N_j=\emptyset$ for $i \neq j$.
Corollary \ref{cor_main} suggests that ${\cal D}^*=\bigcup_{i=1}^kD_{opt}(N_i)$ 
is a liar's dominating set of $G$.  The approximation bound follows from 
Lemma \ref{4-sc}, and Lemma \ref{lem_ptas}. 
Let $|N_i|=n_i$ for $1 \leq i \leq k$. By Lemma \ref{rtime}, an optimal 
liar's dominating set $D_{opt}(N_i)$ of $N_i$ can be computed 
in $n_i^{O(\frac{1}{\epsilon^2}\log \frac{1}{\epsilon})}$ time. Therefore, 
the total running time to compute ${\cal D}^*$ is 
$\sum_{i=1}^kn_i^{O(\frac{1}{\epsilon^2}\log 
\frac{1}{\epsilon})}\leq n^{O(\frac{1}{\epsilon^2}\log \frac{1}{\epsilon})}$.
\end{proof}

\section{Conclusion}\label{conclusion}
In this article, we studied the minimum liar's dominating set problem (MLDS) in UDGs. We proved that the decision version of the MLDS problem is NP-complete. We proposed a simple 7.31-factor approximation algorithm and a PTAS for the problem.
We believe that it is possible to get much better approximation ratios by exploring inherent geometric properties of UDGs.
As a future direction, we work on it and hope to design such algorithms for the problem.
\bibliographystyle{abbrv}
\bibliography{Cocoon_Journal}

\begin{thebibliography}{10}

\bibitem{alimadadi}
A.~Alimadadi, M.~Chellali, and D.~A. Mojdeh.
\newblock Liar’s dominating sets in graphs.
\newblock {\em Discrete Applied Mathematics}, 211:204--210, 2016.

\bibitem{bhore}
S.~Banerjee and S.~Bhore.
\newblock Algorithm and hardness results on liar’s dominating set and
  $k$-tuple dominating set.
\newblock In {\em International Workshop on Combinatorial Algorithms}, pages
  48--60. Springer, 2019.

\bibitem{biedl1998}
T.~Biedl and G.~Kant.
\newblock A better heuristic for orthogonal graph drawings.
\newblock {\em Computational Geometry}, 9(3):159--180, 1998.

\bibitem{garey}
M.~R. Garey and D.~S. Johnson.
\newblock {\em Computers and intractability: a guide to the theory of
  {NP}-completeness}.
\newblock Freeman, 1979.

\bibitem{hopcroft}
J.~Hopcroft and R.~Tarjan.
\newblock Efficient planarity testing.
\newblock {\em Journal of the ACM (JACM)}, 21(4):549--568, 1974.

\bibitem{itai}
A.~Itai, C.~H. Papadimitriou, and J.~L. Szwarcfiter.
\newblock Hamilton paths in grid graphs.
\newblock {\em SIAM Journal on Computing}, 11(4):676--686, 1982.

\bibitem{jallu2017liar}
R.~K. Jallu and G.~K. Das.
\newblock Liar's domination in 2{D}.
\newblock In {\em Conference on Algorithms and Discrete Applied Mathematics},
  pages 219--229. Springer, 2017.

\bibitem{nieberg}
T.~Nieberg and J.~Hurink.
\newblock A {PTAS} for the minimum dominating set problem in unit disk graphs.
\newblock In {\em International Workshop on Approximation and Online
  Algorithms}, pages 296--306. Springer, 2005.

\bibitem{panda2013connected}
B.~S. Panda and S.~Paul.
\newblock Connected liar's domination in graphs: Complexity and algorithms.
\newblock {\em Discrete Mathematics, Algorithms and Applications},
  5(04):1350024, 2013.

\bibitem{panda2013liar}
B.~S. Panda and S.~Paul.
\newblock Liar's domination in graphs: Complexity and algorithm.
\newblock {\em Discrete Applied Mathematics}, 161(7):1085--1092, 2013.

\bibitem{paul2013}
B.~S. Panda and S.~Paul.
\newblock A linear time algorithm for liar's domination problem in proper
  interval graphs.
\newblock {\em Information Processing Letters}, 113(19):815--822, 2013.

\bibitem{panda2014hardness}
B.~S. Panda and S.~Paul.
\newblock Hardness results and approximation algorithm for total liar's
  domination in graphs.
\newblock {\em Journal of Combinatorial Optimization}, 27(4):643--662, 2014.

\bibitem{panda2015hardness}
B.~S. Panda, S.~Paul, and D.~Pradhan.
\newblock Hardness results, approximation and exact algorithms for liar's
  domination problem in graphs.
\newblock {\em Theoretical Computer Science}, 573:26--42, 2015.

\bibitem{roden}
M.~L. Roden and P.~J. Slater.
\newblock Liar's domination in graphs.
\newblock {\em Discrete mathematics}, 309(19):5884--5890, 2009.

\bibitem{shang2008}
W.~Shang, F.~Yao, P.~Wan, and X.~Hu.
\newblock On minimum m-connected k-dominating set problem in unit disc graphs.
\newblock {\em Journal of combinatorial optimization}, 16(2):99--106, 2008.

\bibitem{slater}
P.~J. Slater.
\newblock Liar's domination.
\newblock {\em Networks}, 54(2):70--74, 2009.

\bibitem{sterling}
C.~Sterling.
\newblock {\em Liar's Domination in Grid Graphs}.
\newblock PhD thesis, East Tennessee State University, 2012.

\bibitem{valiant}
L.~G. Valiant.
\newblock Universality considerations in {VLSI} circuits.
\newblock {\em IEEE Transactions on Computers}, 100(2):135--140, 1981.

\end{thebibliography}

\end{document}